\documentclass[final,twocolumn]{IEEEtran}
\usepackage{float}
\usepackage{amssymb,amsmath,amsthm,bm}
\usepackage{mathtools} 
\usepackage{datetime}

\makeatletter
\newcommand{\tpmod}[1]{{\@displayfalse\pmod{#1}}}
\newcommand{\ord}{\mathop{\mathrm{ord}}}
\makeatother

\usepackage{graphicx,color}
\graphicspath{{figures-pdf/}}
\usepackage{cite}
\usepackage{url}

\usepackage{diagbox}
\usepackage[aboveskip=1pt]{subcaption}
\usepackage{multirow,bigstrut}
\usepackage{tabularx}
\usepackage{arydshln}

\usepackage{tocloft}                    
\usepackage[mathlines]{lineno}

\newlength\OneImW
\newlength\TwoImW
\newlength\BigOneImW
\newlength\twofigwidth
\newlength\sfigwidth
\newlength\vfigskip
\newtheorem{theorem}{Theorem}
\newtheorem{proposition}{Proposition}
\newtheorem{corollary}{Corollary}
\newtheorem{lemma}{Lemma}

\newcommand{\Zy}[1]{\mathbb{Z}_{3^{#1}}}
\newcommand{\Zt}[1]{\mathbb{Z}_{2^{#1}}}

\newcommand{\bnu}{\boldsymbol{\nu}}
\newcommand{\vvert}[0]{{\;\vert\;}}

\usepackage{hyperref}

\renewcommand\baselinestretch{1}

\begin{document}


\title{Graph Structure of Chebyshev Permutation Polynomials over Binary and Ternary Adic Rings}

\author{Xiaoxiong Lu, Yuling Dai, Chengqing Li
\thanks{This work was supported by the National Natural Science Foundation of China (no.~62571470).}

\thanks{X. Lu is with the College of Mathematics and Statistics, Hengyang Normal University, Hengyang, 421008, China.}

\thanks{Y. Dai and C. Li are with the School of Computer Science, Xiangtan University, Xiangtan 411105, Hunan, China.}
}
\markboth{IEEE Transactions}{Li\MakeLowercase{\textit{et al.}}}

    \IEEEpubid{\begin{minipage}{\textwidth}\ \\[12pt] \centering
    \\[2\baselineskip]
    0018-9448 \copyright 2025 IEEE. Personal use is permitted, but republication/redistribution requires IEEE permission.\\
      See http://www.ieee.org/publications\_standards/publications/rights/index.html for more information.\\ \today{}\space(\currenttime)
    \end{minipage}
}
   
\maketitle

\begin{abstract}
Understanding the functional graph of a nonlinear map over a finite domain is crucial for analyzing its dynamical complexity and potential applications in cryptography and pseudorandom generation. In this paper, we investigate the graph structure of Chebyshev permutation polynomials over the ring $\mathbb{Z}_{2^{k_1}3^{k_2}}$, where $k_1$ and $k_2$ are positive integers
and $0\in\{k_1, k_2\}$.
Each element of the ring is regarded as a vertex, and the mapping relation defined by the polynomial corresponds to a directed edge. Building on new properties of Chebyshev polynomials modulo powers of $2$ and $3$, we provide an explicit characterization of path lengths and cycle structures in the functional graph. We show that, despite the complexities introduced by the binary and ternary components, the graph exhibits strong regularities, including a constant number of cycles of a given length and predictable branching patterns as $k_1$ and $k_2$ increase. Our results extend previous studies over prime-power rings, offering insights into the emergence of complexity in digital nonlinear maps and supporting the security analysis of their cryptographic applications.
\end{abstract}



\section{Introduction}

\IEEEPARstart{D}igital computing environments are inherently constrained by discrete state spaces and finite word lengths. 
When nonlinear systems defined on continuous domains, typically represented by the Logistic map, are implemented on digital hardware, their original complex dynamical characteristics inevitably degrade due to finite-precision effects in floating-point or fixed-point arithmetic.
This degradation manifests as periodic outputs and statistical deviations from ideal chaotic behavior, ultimately compromising security and the quality of randomness \cite{cqli:network:TCASI2019}. 
To mitigate dynamical degradation, prior studies employ external interventions to sustain system complexity. Typical strategies include paralleling multiple nonlinear systems~\cite{TXJ:imageEncryption:2014}, injecting random perturbations into generated sequences \cite{wang:perturbation:TCAS2016}, and dynamically controlling parameters~\cite{Tan:random:ND2024}.

However, compared to these remedial strategies, a more fundamental approach is to forgo approximating continuous systems and instead reconstruct nonlinear maps over discrete algebraic structures, such as finite fields or finite rings~\cite{chenf:cat2:TIT13,Liao:ITC:2010,licq:Logistic:IJBC2023}.
Discrete nonlinear systems inherently preclude the dynamical degradation induced by finite-precision computation due to their finite state spaces. Furthermore, these systems facilitate exact graph-theoretic and algebraic analyses of key dynamical features, including global period distributions and local evolutionary behaviors. Notably, the residue class ring $\mathbb{Z}_m$ is predominantly employed for such analyses owing to its structural naturalness.

Notably, when the modulus is chosen as $m=2^k$, fixed-point integers are naturally represented as $k$-bit elements in the ring $\mathbb{Z}_{2^k}$. This algebraic structure directly corresponds to the arithmetic logic of $k$-bit registers in modern digital computers, where addition and multiplication operations naturally satisfy the closure property modulo $2^k$. Compared to executing modular arithmetic over the finite field $\mathbb{Z}_p$, nonlinear iterations over $\mathbb{Z}_{2^k}$ leverage the hardware-level natural overflow mechanism to achieve extremely high computational efficiency~\cite{Yoshioka:isit:2023}. 
Therefore, discrete nonlinear systems constructed over the ring $\mathbb{Z}_{2^k}$, such as various linear and nonlinear congruential generators, constitute the cornerstone of pseudorandom number generation in modern operating systems and standard libraries.

\IEEEpubidadjcol 

Beyond the mainstream binary architecture, the scenario in which the modulus is chosen as $m=3^k$ also has profound scientific significance. Examining from the fundamental dimension of information theory, the hardware complexity required to store an integer $N$ is proportional to $r \cdot \log_r N$, where $r$ denotes the radix. The extremum of this function is located at $e \approx 2.718$, and among all integers, $3$ is the closest to $e$. This implies that, at the theoretical level, ternary logic exhibits higher storage density and a simpler interconnection topology than binary logic~\cite{knuth1997vol2}.
In this paper, we call finite binary and ternary adic rings as
$\mathbb{Z}_{2^k}$ and $\mathbb{Z}_{3^k}$, respectively.

Chebyshev polynomials can provide near-optimal polynomial approximations, forming the foundation for stable and rapidly convergent methods in function approximation, interpolation, and the numerical solution of linear systems.
The fundamental recurrence relation of Chebyshev polynomials can be expressed as
$T_n(x) = 2xT_{n-1}(x) - T_{n-2}(x).$
Their inherent semigroup property serves as a fundamental mechanism for their widespread applications, which is formulated as
$
T_n(T_m(x)) = T_{nm}(x).
$
With the ongoing evolution of nonlinear dynamics and digital cryptographic systems, Chebyshev polynomials occupy an irreplaceable central position, particularly in public-key encryption mechanisms and pseudorandom sequence design~\cite{kocarev:2003:public}.

Early relevant discussions predominantly focused on continuous Chebyshev polynomials over real or complex fields. Specifically, when $x \in [-1, 1]$, this polynomial can be explicitly defined through inverse trigonometric functions, and its analytical expression is equivalently represented as
$$
T_n(x) = \cos(n \arccos x).
$$
However, cryptanalytic techniques such as phase space reconstruction, chosen-plaintext attacks, and orbital parameter deciphering have compromised various continuous-domain Chebyshev encryption protocols~\cite{Bergamo:PKIchebyshev:CASI2005,Chen:chebyZN:IS2011}. These vulnerabilities highlight the inherent theoretical and security limitations of continuous chaotic systems in physical implementations.

To address the inherent finite-precision effects of continuous chaotic systems in engineering implementations, the academic community has gradually shifted its research focus toward discrete algebraic systems. 
Researchers have extensively explored Chebyshev polynomials defined over finite fields $\mathbb{F}_N$ and residue class rings $\mathbb{Z}_N$~\cite{cqli:Cheby:TIT25,Daniel:Chebyshev:2018,Liao:ITC:2010,Yoshioka:ChebyshevPk:TCAS2:2018}, where $N$ is an integer. Under these conditions, the Chebyshev polynomial is formulated as
\begin{equation}\label{chevby:express:fild}
T_n(x) = 2xT_{n-1}(x) - T_{n-2}(x) \bmod N.
\end{equation}
In these discrete algebraic structures, Chebyshev polynomials not only perfectly preserve the original semigroup property but also eliminate the precision loss issues caused by floating-point operations~\cite{Liao:ITC:2010}. 
Particularly over the residue class ring of prime powers $\mathbb{Z}_{p^k}$, Chebyshev permutation polynomials exhibit more complex algebraic structures and exceptionally superior cryptographic properties~\cite{Chen:chebyZN:IS2011}. 
When the degree $n$ of the polynomial and the characteristic of the residue ring satisfy specific coprime conditions, the Chebyshev polynomial constitutes a strict permutation polynomial. 
This indicates that the mapping $f(x) \equiv T_n(x) \pmod{p^k}$ over the residue class ring $\mathbb{Z}_{p^k}$ possesses a perfect global bijective property.

Despite mitigating dynamical degradation, discrete Chebyshev maps remain vulnerable~\cite{Yoshioka:ChebyshevPk:TCAS2:2020}. Improper choices of the polynomial degree $n$ or the initial state $x_0$ lead to extremely short periods or concentrated fixed points, exposing topological vulnerabilities. Consequently, functional graph analysis has emerged as a core methodology for screening long-period orbits and ensuring the security of discrete Chebyshev cryptosystems. Furthermore, extensive theoretical studies have quantitatively characterized the permutation rules, period distributions, and trajectory lengths of Chebyshev polynomials across various discrete algebraic structures.

As illustrated in Table~\ref{tab:CombinedPeriods}, Daniel et al. systematically provided the criteria for the period $T$ of Chebyshev polynomials over the finite field $\mathbb{F}_q$ and its corresponding effective number of orbits $N_T$~\cite{Daniel:Chebyshev:2018}. Breaking through the limitations of prime fields, Li et al. further derived more complex periodic structures and counting formulas over the prime power residue class ring $\mathbb{Z}_{p^k}$~\cite{cqli:Cheby:TIT25}. 
Meanwhile, for the prime residue ring $\mathbb{Z}_N$, Liao et al. detailed the mapping relationship between the reducibility of the characteristic polynomial $f(t)=t^{2}-2xt+1 \in \mathbb{Z}_{N}[t]$ and the corresponding period $T'$~\cite{Liao:ITC:2010}. 
To clarify the notations utilized in Table \ref{tab:CombinedPeriods}, the corresponding symbols are explicitly defined. The variables $T$ and $T'$ represent the minimum positive integers satisfying $T^{s}_n(x)=T^{s+T}_n(x)$ for any positive integer $s$ and element $x$ in the corresponding domain, and $T_n(x)=T_{n+T'}(x)$ for any positive integer $n$ and element $x$, respectively. The symbol $\varphi$ denotes Euler's totient function, $N$ indicates a positive prime, and $\tilde{o}_d(n)$ characterizes the order of $n$ within the quotient group $\mathbb{Z}^*_N/\{1, -1\}$. Furthermore, the auxiliary parameters are determined by the mathematical expressions $N_s = \min\{i\vvert T^{i}_n(x)\equiv x_0\pmod{p}\}$, $l_s = \ord \left( T'_{n^{N_s}}(x_0) \right)$, and $w = \bnu_p(n^{2 \cdot \ord(n^2)} - 1)$.

Recently, Panraksa and Tangboonduangjit \cite{panraksa:fixedChebysev:arxiv2026} investigated the non-permutation dynamical behavior of Chebyshev polynomials over $\mathbb{Z}_{p^k}$. 
However, this theoretical framework fails to encompass the cases of the rings $\mathbb{Z}_{2^k}$ and $\mathbb{Z}_{3^k}$. 
For $p=2$, although Yoshioka\cite{Yoshioka:TCSII:2016} explored the periodic properties of the sequences, their work did not provide explicit expressions for the periods nor a precise characterization of the points from a global topological perspective. Given the limitations of conventional methods in these low-order cases, this paper bridges this theoretical gap by treating $\mathbb{Z}_{2^k}$ and $\mathbb{Z}_{3^k}$ as independent algebraic structures and systematically analyzing them separately. 
Moreover, this paper derives the explicit expressions for the periods and the precise characterization of points over the ring $\mathbb{Z}_{2^k}$, and thoroughly reveals the complete functional graph structure and period distribution over the ring $\mathbb{Z}_{3^k}$.


\begin{table}[!htb]
\caption{Conclusions regarding the period distribution of sequences generated by Chebyshev polynomials in the literature \cite{Daniel:Chebyshev:2018}, \cite{cqli:Cheby:TIT25}, and \cite{Liao:ITC:2010}}
\centering
\resizebox{\linewidth}{!}{
\begin{tabular}{*{3}{c|}c}
    \hline
    Reference & $T$ or $T'$ & Condition & $N_T$ or $N_{T'}$ \\
    \hline    
    \multirow{3}{*}{\cite{Daniel:Chebyshev:2018}}
     & $\tilde{o}_d(n)$ & $d>2$ and $(d\mid(q-1)$ or $d\mid(q+1))$ & $\frac{\varphi(d)}{2\tilde{o}_d(n)}$\\ \cline{2-4}
     &
    \multirow{2}{*}{1}  & $q$ is an odd number & 1 \\ 
    \cline{3-4}
     &  & $q$ is an even number & 2 \\ \hline

     \multirow{4}{*}{\cite{cqli:Cheby:TIT25}}
     & $N_s$ & $k\geq1$ & $1$ \\\cline{2-4}
     &\multirow{2}{*}{$N_s\cdot l_s$}
        & $w+1\geq k\geq1$ & $\frac{p^{k-1}-1}{l_s}$ \\ \cline{3-4}
        & & $k> w+1$ & $\frac{p^w-1}{l_s}$ \\
    \cline{2-4}
    &$N_sl_sp^t$  & $k> w+1$ & $\frac{(p-1)p^{w-1}}{l_s}$ \\ \hline

    \multirow{4}{*}{\cite{Liao:ITC:2010}}
     & $1$ & $x=1$ & $1$ \\\cline{2-4}
     & $2$ & $x=N-1$ & $1$ \\ \cline{2-4}
    & $T' \in \{ d : d \mid N-1, d>2 \}$  & $x\neq\pm1$, $f(t)$ is reducible & \multirow{2}{*}{$\frac{\varphi(T')}{2}$} \\
    \cline{2-3}
    & $T' \in \{ d : d \mid N+1, d>2 \}$  & $x\neq\pm1$, $f(t)$ is irreducible \\
    \hline
\end{tabular}
}
\label{tab:CombinedPeriods}
\end{table}

\section{Preliminaries}

To make the analysis on the periodic change process of the Chebyshev polynomial sequence complete, the previously known results about the Chebyshev polynomial are briefly introduced in this section.

Chebyshev polynomials of degree $n$ are recursively defined by
\begin{equation}
    T_{n}(x)=2xT_{n-1}(x)-T_{n-2}(x),
\end{equation}
where $T_0(x)=1$, $T_1(x)=x$, $x\in [-1, 1]$ and $n$ is an integer lager than one.
The power series form of Chebyshev polynomials of degree $n$ can be represented as
\begin{equation}\label{eq:cheby:power}
    T_n(x)=\frac{n}{2} \sum_{i=0}^{\lfloor n/2 \rfloor} (-1)^i \frac{(n-i-1)!}{i! (n-2i)!} (2x)^{n-2i},
\end{equation}
where $n\ge 1$ \cite{bulychev:Chebyshev:ARC2003}.

Chebyshev polynomials exhibit a significant semigroup property. 
When these polynomials are evaluated over the residue ring $\mathbb{Z}_{p^k}$, this property is preserved, meaning that for any two positive integers $m$ and $n$, one has
\begin{equation*}
    T_n(T_m(x)) \equiv T_m(T_n(x)) \equiv T_{m \cdot n}(x) \pmod{p^k}.
\end{equation*}
Thus, the composition of $N$ Chebyshev polynomials $T_n(x)$ can be denoted $T_{n^{N}}(x)$ or $T^{N}_{n}(x)$.

A Chebyshev integer sequence $S_{p^k}(x_0; n)$ is generated by iterating the Chebyshev polynomial over $\mathbb{Z}_{p^{k}}$.
Namely, given an initial value $x$, the result of $i$-th iteration is
\begin{equation*}
T^i_n(x)=T_n(T^{i-1}_n(x))\bmod p^k,
\end{equation*}
where $i\geq 1$ and $T^0_n(x)=x\bmod p^k$. The least period $L(n, p, k, x)$ of the sequence $S_{p^k}(x_0; n)$ is the minimum possible positive integer of $N$ satisfying
$T^N_n(x)\equiv x \bmod p^k$. 
The Chebyshev polynomial $T_n(x)$ is a permutation polynomial modulo $p^k$ if and only if
\begin{equation}\label{PermuteCondition}
    \gcd(n, p)=\gcd(n, p^2-1)=1,
\end{equation}
where $\gcd$ denotes the greatest common divisor of two integers~\cite{Yoshioka:ChebyshevPk:TCAS2:2018}. 
For $p=2$, the necessary and sufficient condition for $T_n(x)$ to be a permutation polynomial simplifies to $n$ being an odd integer~\cite{Umeno:Permute}. 
Similarly, for $p=3$, the condition reduces to $n \equiv \pm1\pmod6$. 

\section{The Period Properties of Chebyshev Integer Sequence}

This section investigates the distribution of periods for sequences generated through the iteration of Chebyshev permutation polynomials over the rings $\mathbb{Z}_{p^k}$, where $p\in\{2, 3\}$.
Given that $n$ is an odd integer, it follows that $2 \mid (n \pm 1)$. 
Defining 
\begin{equation*}
    w = \max\{\bnu_p(n-1), \bnu_p(n+1)\}
\end{equation*}
for $p \in \{2, 3\}$, the inequality $w \geq 2$ holds for the specific case where $p=2$ and $n \geq 2$. 
Consequently, the $p$-adic valuation of $n^2-1$ satisfies $\bnu_2(n^2-1) \geq w+1$ for $p=2$, and the exact relation $\bnu_3(n^2-1) = w$ is established for $p=3$.

\subsection{The Period Properties of Chebyshev Integer Sequence over Ring $\mathbb{Z}_{2^k}$}

Lemma~\ref{pro:2w:Tn'1} and Lemma~\ref{pro:2w:ak} characterize the multiplicity of the factor $2$ within the derivatives at the endpoints and the coefficients of the polynomials. Building upon these results, Theorem~\ref{theorem:p2:explan} illustrates the evolution of sequence periods.
Congruence
\begin{equation}\label{eq:s2}
T_n(x_0) \equiv x_0 \pmod{2^2}
\end{equation}
can be verified as follows. When $x_0 \equiv 1 \pmod{4}$, the identity $T_n(1) = 1$ immediately implies that congruence~\eqref{eq:s2} holds. When $x_0 \equiv 2 \pmod{4}$, since $T_1(2) = 2$ and the recurrence relation $T_n(2) \equiv -T_{n-2}(2) \pmod{4}$, it follows that congruence~\eqref{eq:s2} holds in this case as well. When $x_0 \equiv 3 \pmod{4}$, the identity $T_n(-1) = (-1)^n = -1$ gives the same result, and congruence~\eqref{eq:s2} holds. For the case $x_0 \equiv 0 \pmod{4}$, the congruence can be confirmed through computational checks. Therefore, $s \geq 2$ in Theorem~\ref{theorem:p2:explan}.
Based on the identified periodic patterns, Theorem~\ref{theorem:p2:period:express} further establishes explicit expressions for the sequences.

\begin{lemma}\label{pro:2w:Tn'1}
As for any Chebyshev polynomial $T_n(x)$, one has $2^{w+1} \mid T'_n(\pm1)-1$
and 
\begin{equation}\label{eq:2wmid:T'n}
2^{w+\lfloor \frac{m}{2} \rfloor} \mid T^{(m)}_n(\pm1)
\end{equation}
for any $m\ge 2$, where $w =\max\{\bnu_2(n-1), \bnu_2(n+1)\}$.
\end{lemma}
\begin{proof}
The $m$-order derivative of Chebyshev polynomials $T_n(x)$ with $x=\pm1$ can be expressed
\begin{equation}\label{eq:T'n:+-1}
T^{(m)}_n(\pm 1) = (\pm1)^{n+m}\prod\limits_{j=0}^{m-1}\frac{n^2-j^2}{2j+1}
\end{equation}
for any $m$,
one has $T'_n(\pm 1) =n^2$.
From the definition of $w$ and $n$ is odd, one has 
\begin{equation}\label{eq:v2(n^2-1)}
\bnu_2(n^2-1) = \bnu_2(n+1) + \bnu_2(n-1) \geq w+1,
\end{equation}
yielding that
\(2^{w+1}\mid (n^2-1).\)
In \eqref{eq:T'n:+-1}, one has $2\mid (n^2-j)$ when $j$ is odd from $n$ is odd.
It means $2^{\lfloor \frac{m}{2}\rfloor-1}\mid
\prod\limits_{j=2}^{m-1}{(n^2-j^2)}$.
It yields from \eqref{eq:T'n:+-1} and \(2^{w+1}\mid (n^2-1)\) that $2^{w+\lfloor \frac{m}{2}\rfloor} \mid T^{(m)}(\pm1)$ for any $m\ge 2$.
\end{proof}

\begin{lemma}\label{pro:2w:ak}
As for any Chebyshev polynomial $T_n(x)$, one has
$2^w\mid (a_1-1)$
and
$2^w\mid a_i$,
where $w =\max\{\bnu_2(n-1), \bnu_2(n+1)\}$.
$a_i$ denotes the coefficient of $x^i$ in the Chebyshev polynomial and $i\in\{3, 5, \cdots, n\}$.
\end{lemma}
\begin{proof}
Set $n-2i=2j+1$ in \eqref{eq:cheby:power}, one has 
\begin{equation}\label{eq:cheby:power:2}
T_n(x)=\frac{n}{2}\sum_{j=0}^{(n-1)/2} (-1)^{\frac{n-1}{2}+j} \frac{\frac{n+2j-1}{2}!}{
\frac{n-2j-1}{2}!(2j+1)!} (2x)^{2j+1},
\end{equation}
where $j\in\{0, 1, \cdots, \frac{n-1}{2}\}.$
Namely,
\begin{equation}\label{eq:cheby:coeff}
a_{2j+1}  = (-1)^{\frac{n-1}{2}+j}\binom{\frac{n-1}{2}+j}{2j}\cdot(2j+1)^{-1}\cdot n \cdot 2^{2j}.
\end{equation}
From the definition of $w$, one has 
$n= \mp 1+ h2^w$ and $w\ge 2$, where $h\in \mathbb{Z}$. 
It means $a_1=(-1)^\frac{n-1}{2}n = 1\mp 2^wh$, where $h$ is an integer.
Thus, $2^w \mid (a_1-1)$. 

Next, it follows from $2j+1$ is odd and \eqref{eq:cheby:coeff} that $\bnu_2(a_{2j+1}) \geq 2j$. The proof is divided into two cases depending on the value of $j$:
\begin{itemize}
\item $2j\geq w$: 
One has $\bnu_2(a_{2j+1}) \geq 2j\ge w$, namely $2^w\mid a_{2j+1}$.
\item $0\leq 2j < w$: 
One can get $\bnu_2(n^2-1)= w+1 \geq 2(j+1)$ from \eqref{eq:v2(n^2-1)} and $w$ is an integer.
It yields from $2(j+1) \geq \bnu_2(2(j+1))$ that 
\begin{equation}\label{eq:leq1}
\bnu_2(n^2-1)\geq \bnu_2(2(j+1)).
\end{equation}
Since
\(
\frac{a_{2j+3}}{a_{2j+1}} = 
-\frac{n^2-(2j+1)^2}{(2j+3)(2j+2)} = \frac{n^2-1-4j(j+1)}{2(j+1)},
\)
one has 
\(
\bnu_2(a_{2j+3})-\bnu_2(a_{2j+1})= 
\min\{\bnu_2(n^2-1), \bnu_2(4j(j+1))\}-\bnu_2(2(j+1)).  
\)
Then, according to \eqref{eq:leq1} and $\bnu_2(4j(j+1))\ge\bnu_2(2(j+1))$, one has 
$$\bnu_2(a_{2j+3})-\bnu_2(a_{2j+1})\geq0.$$
Note that 
$a_{3} = (-1)^{\frac{n-1}{2}+1}\frac{n(n+1)(n-1)}{3\cdot2}
$, 
one has  $\bnu_2(a_3)\ge w$ from \eqref{eq:v2(n^2-1)}.
Thus, one can get $\bnu_2(a_{2j+3})\ge \bnu_2(a_{2j+1}) \ge \bnu_2(a_3)\ge w$ and $2^w\mid a_i$.
\end{itemize}
\end{proof}

\begin{theorem}\label{theorem:p2:explan}
If a Chebyshev permutation polynomial with an initial value $x_0\in \mathbb{Z}_{2^k}$ satisfies
\begin{equation}\label{eq:begin:explan:p2}
    \begin{cases}
    T_n(x_0)\equiv x_0\pmod{2^s}; \\
    T_n(x_0)\not\equiv x_0\pmod{2^{s+1}},
    \end{cases}
\end{equation}
then
\begin{equation}\label{eq:2ks}
   T_n^{2^{k-s}}(x_0)\equiv x_0\pmod{2^k}, 
\end{equation}
and the least period of the sequence $S_{2^k}(x_0; n)$ is $2^{k-s}$, 
where $k\geq s$ and $v\geq 0$.
\end{theorem}
\begin{proof}
The Taylor expansion of a polynomial $f(x)$ at a point $x$ is given by
\begin{equation}\label{eq:TaylorExplan}
 f(x+h) = f(x) + h \cdot f'(x) + \sum^{\deg(f)}_{j=2} \frac{f^{(j)}(x)}{j!} \cdot h^j,
\end{equation}
where $\deg(f)$ denotes the degree of the polynomial $f(x)$.

Next, mathematical induction on $t$, one can prove
\begin{equation}\label{eq:relation:period:p2}
\begin{cases}
T_n^{2^{t}}(x_0)\equiv x_0\pmod{2^{s+t}}; \\
T_n^{2^{t}}(x_0)\not\equiv x_0\pmod{2^{s+t+1}},
\end{cases}
\end{equation}
where $t$ is a non-negative integer.
The base case of induction, $t = 0$, is directly given by the conditions in Eq.~\eqref{eq:begin:explan:p2}, which state that
\[
\begin{cases}
    T_n(x_0) \equiv x_0 \pmod{2^s}, \\
    T_n(x_0) \not\equiv x_0 \pmod{2^{s+1}}.
\end{cases}
\]
Assuming that this pair of modular relations holds for $t = m$, namely
\[
\begin{cases}
 T_n^{2^m}(x_0) \equiv x_0 \pmod{2^{s+m}}, \\
 T_n^{2^m}(x_0) \not\equiv x_0 \pmod{2^{s+m+1}},
\end{cases}
\]
the goal is to demonstrate that \eqref{eq:relation:period:p2} holds for $t = m+1$.

Let $\Phi(x_0) = T_n^{2^m}(x_0)$. From the inductive hypothesis, it follows that
\[
 \Phi(x_0) = x_0 + q_m \cdot 2^{s+m},
\]
where $q_m$ is an odd integer. The next iteration of $T_n$, denoted $T_n^{2^{m+1}}(x_0)$, is given by
\[
 T_n^{2^{m+1}}(x_0) = T_n^{2^m}(T_n^{2^m}(x_0)) = \Phi(\Phi(x_0)) = \Phi(x_0 + q_m \cdot 2^{s+m}).
\]

Now, the Taylor expansion \eqref{eq:TaylorExplan} is applied with $x = x_0$ and $h = q_m \cdot 2^{s+m}$. Expanding $\Phi(x_0 + q_m \cdot 2^{s+m})$ yields
\begin{multline*}
\Phi(x_0 + q_m \cdot 2^{s+m}) = \Phi(x_0) + q_m \cdot 2^{s+m} (\Phi'(x_0)+1) \\
+ \sum_{j=2}^{\deg(\Phi)} \frac{\Phi^{(j)}(x_0)}{j!} (q_m \cdot 2^{s+m})^j.   
\end{multline*}
The derivative of $\Phi(x_0)$ is given by
\[
 \Phi'(x_0) = \prod_{i=0}^{2^m-1} T'_n(x_i),
\]
and since $T'_n(x_i) \equiv 1 \pmod{4}$, it follows that
\[
 \Phi'(x_0) \equiv 1 \pmod{4}.
\]
Since $T_n(x_0)\equiv x_0\pmod4$, it follows from~\eqref{eq:begin:explan:p2} that $s\geq2$. 
Combining with the condition $m\geq1$, then $\bnu_2\left(\frac{\Phi^{(j)}(x_0)}{j!} (q_m 2^{s+m})^j\right)\geq 1+(s+m)j-j+s_2(j)\geq2(s+m)\geq s+m+2$. 
Thus, $\Phi(x_0 + q_m \cdot 2^{s+m})=T_n^{2^{m+1}}(x_0) \equiv x_0 + q_m \cdot 2^{s+m+1}\pmod{2^{s+m+2}}$. 
This means \eqref{eq:relation:period:p2} remains valid for $t = m+1$.

To conclude, this result is applied to the specific case where $t = k - s$ and $t = k - s - 1$. From the induction, one deduces
\[
\begin{cases}
T_n^{2^{k-s}}(x_0) \equiv x_0 \pmod{2^k}, \\
T_n^{2^{k-s}}(x_0) \not\equiv x_0 \pmod{2^{k+1}},
\end{cases}
\]
and
\[
\begin{cases}
T_n^{2^{k-s-1}}(x_0) \equiv x_0 \pmod{2^{k-1}}, \\
T_n^{2^{k-s-1}}(x_0) \not\equiv x_0 \pmod{2^k}.
\end{cases}
\]
Thus, the least period of the sequence $S_{2^k}(x_0; n)$ is $2^{k-s}$, and the least period of $S_{2^{k+v}}(x_0; n)$ is $2^v$ times that of $S_{2^k}(x_0; n)$.
\end{proof}

\subsection{The Period Properties of Chebyshev Integer Sequence over Ring $\mathbb{Z}_{3^k}$}

Lemma~\ref{lemma:p3:Tn'1} provides an exact estimation of the power of 3 dividing the derivatives of Chebyshev polynomials at the endpoints. Additionally, Lemma \ref{lemma:dev:m:0mod3:w+2} characterizes the congruence properties of the first derivative at the origin and evaluates their influence on the divisibility of higher-order derivatives. Finally, Theorem \ref{Theorem:periodValue} establishes an explicit expression for the period distribution of Chebyshev sequences.

\begin{lemma}\label{lemma:p3:Tn'1}
As for any Chebyshev polynomial $T_n(x)$, one has 
\begin{equation*}\label{eq:3w:T'n}
\bnu_3\left(\frac{T^{(m)}_n(\pm 1) \cdot 3^m}{m!}\right)\ge w+2
\end{equation*}
for any $m\ge 3$.
where $w$ is a positive integer.
\end{lemma}
\begin{proof}
The validity of the Lemma for $3 \leq m \leq 5$ is verified through direct calculation using \eqref{eq:T'n:+-1}. Consequently, the subsequent proof concentrates on the case $m \geq 6$.

Referring to \eqref{eq:T'n:+-1} and $\bnu_3(n^2-1)=w$, one obtains
$\bnu_3\left(\frac{T^{(m)}_n(\pm 1) \cdot 3^m}{m!}\right) \ge w+\bnu_3(B)-\bnu_3((2m-1)!!)+m-\bnu_3(m!)$, 
where $B=\prod\limits_{j=2}^{m-1}(n^2-j^2)$.
Note that
$\bnu_3(m!)+\bnu_3((2m-1)!!)=\bnu_3(m)+\bnu_3((2m-1)!)$ from $(2m-1)!!=\frac{m\cdot(2m-1)!}{2^{m-1} \cdot m!}$.
Thus, one has 
$\bnu_3\left(\frac{T^{(m)}_n(\pm 1) \cdot 3^m}{m!}\right) \ge w+\bnu_3(B)-\bnu_3((2m-1)!)+m-\bnu_3(m).$
According to Legendre's formula, 
$\bnu_3(x!)=\frac{x-s_3(x)}{2}$, 
one has $\bnu_3((2m-1)!)=\frac{2m-1-s_3(2m-1)}{2}$.
As $s_3(x)\ge 1$ for $x>0$, it further deduces
$\bnu_3((2m-1)!)\leq m-1 $. 
Hence, 
\begin{equation}\label{v_3:Tm:m!}
\begin{aligned}
\bnu_3\left(\frac{T^{(m)}_n(\pm 1) \cdot 3^m}{m!}\right)
\geq w+1+\bnu_3(B)-\bnu_3(m).
\end{aligned}
\end{equation}
Let $\bnu_3(m)=h$,  one has $m=a3^h$, where $a\geq 1$ and $h\ge 0$.
From $n^2-1\equiv 0 \pmod {3^w}$, one has $n^2-j^2\equiv 0\pmod 3$ if $j \equiv \pm 1 \pmod 3$.
Thus, one gets 
\begin{equation}\label{T1:rela:m!}
\bnu_3(B)=\lfloor \frac{2m-3}{3}\rfloor = 
\lfloor \frac{2\cdot a3^h-3}{3}\rfloor.
\end{equation}
If $h=0$, one has $\lfloor\frac{2m-3}{3}\rfloor\geq h+1$ from $m\ge 3$.
If $h=1$, one has $a \geq 2$ from $m\ge 6$, and $\lfloor\frac{2\cdot a3^h-3}{3}\rfloor=2a - 1 \geq h+1$.
If $h>1$, one has the inequality $2a \cdot 3^{h-1} - 1 \geq h+1$.
Thus, one can obtains $\bnu_3(B)\geq \bnu_3(m)+1$.
Then, it further gets
\(
\bnu_3\left(\frac{T^{(m)}_n(\pm 1) \cdot 3^m}{m!}\right) \geq w+2
\)
for any $m \geq 6$ from \eqref{v_3:Tm:m!}.
\end{proof}

\begin{lemma}\label{lemma:dev:m:0mod3:w+2}
Let \(
w =\max\{\nu_3(n-1), \nu_3(n+1)\}
\),
one has $T'_{n^{l_0}}(0) \equiv 1\pmod {3^w}$
and 
\(
\bnu_3\left(\frac{T^{(m)}_{n^{l_0}}(0) \cdot 3^m}{m!}\right)\ge w+2
\)
for any $m\ge2$, where $l_0=\ord(T'_n(0))$.
\end{lemma}
\begin{proof}
From the definition of $w$, one has $n\equiv\pm1\pmod{3^w}$.
Note that the case $n\equiv-1\pmod{3^w}$ is similar to $n\equiv1\pmod{3^w}$, only the latter case is considered in the following proof.
 
Setting $n=1+h\cdot3^w$, one has  $h\bmod4\in\{0, 2\}$ as $n$ is odd.
It yields from~\eqref{eq:cheby:power} that $T'_n(0)=(-1)^{\frac{h\cdot3^w}{2}}n$.
Next, one can prove $T'_{n^{l_0}}(0)\equiv1\pmod{3^w}$ depending on the value of $h$:
\begin{itemize}
\item  $h\equiv0\pmod4$: One has 
$T'_n(0)=n$, and $n\equiv1\pmod3$ from $n\equiv1\pmod{3^w}$.
So, $l_0=\ord(T'_n(0))=1$, which means $T'_{n^{l_0}}(0)=n\equiv1\pmod{3^w}$. 
\item $h\equiv2\pmod4$: 
One has $T'_n(0)=-n$, and $-n\equiv-1\pmod3$ from $n\equiv1\pmod{3^w}$.
So, $l_0=\ord(T'_n(0))=2$.
So, $T'_{n^2}(0)=(-1)^{\frac{n^2-1}{2}}n^2 =n^2$. 
It yields from $n\equiv1\pmod{3^w}$ that $T'_{n^2}(0)\equiv1\pmod{3^w}$.
\end{itemize}
Hence, $T'_{n^{l_0}}(0)=n\equiv1\pmod{3^w}$. 
It means $n^{l_0}\equiv1\pmod{3^w}$ and $n^{2l_0}\equiv1\pmod{3^w}$.
Furthermore, from \cite[Lemma 3]{cqli:Cheby:TIT25}, one has $T^{(m)}_{n^{l_0}}(0) = -({n^{2l_0}} - (m-2)^2) T^{(m-2)}_{n^{l_0}}(0)$.
Then one can get $\bnu_3(T^{(m)}_{n^{l_0}}(0)) \ge \bnu_3(n^{2l_0} - 1) \ge w$ for $m\equiv1\pmod2$ and $T^{(m)}_{n^{l_0}}(0)=0$ for $m\equiv0\pmod2$. Hence, $\bnu_3(T^{(m)}_{n^{l_0}}(0))\ge w$.
Combining with $\bnu_3(\frac{3^m}{m!})\ge 2$, one can get \(
\bnu_3\left(\frac{T^{(m)}_{n^{l_0}}(0) \cdot 3^m}{m!}\right)\ge w+2.
\)
\end{proof}

\begin{theorem}\label{Theorem:periodValue}
Given $s\in\mathcal{Z}_{p^k}\backslash \{0,1\}$, the least period of sequence $S_{3^k}(s; n)$ is $l_s\cdot 3^{k-v_s}$ when $k\ge v_s$, where $l_s=\ord(T'_n(s))$ and $v_s=\bnu_p(T_n^{l_s}(s)-s)$.
\end{theorem}

\begin{proof}
First, based on the different values of $s$, it can be proven that $v_s \geq 2$.
\begin{itemize}
\item $s\equiv\pm1\pmod3$:
In this scenario, $s$ can be expressed as $s=3h \pm 1$ for some integer $h$. Combining the definition of $n$ yields $l_s=1$. 
Applying Taylor's formula provides
$$
T_n(\pm1+3\cdot h)=T_n(\pm1)+T'_n(\pm1)\cdot3h+\sum^n_{i=2}\frac{T^{(i)}_n(\pm1)3^ih^i}{i!}.
$$
It yields from $\frac{T^{(i)}_n(\pm1)}{i!}$ is an integer that $\bnu_3\left(\frac{T^{(i)}_n(\pm1)3^ih^i}{i!}\right)\geq2$. 
Thus, $T_n(\pm1+3\cdot h)\equiv\pm1+3hT'_n(\pm1)\pmod{3^2}$. 
Furthermore, from $T'_n(\pm1)\equiv n^2\equiv1\pmod{3^2}$, one has
$T_n(\pm1+3\cdot h)\equiv\pm1+3h\pmod{3^2}$. 
Based on the definition of $v_s$, the inequality $v_s\geq2$ is established.

\item $s\equiv0\pmod3$:
One can get $s=3h$ and $l_s\in\{1, 2\}$. 
Analogous to the preceding case, applying Taylor's formula, one has $T_{n^{l_s}}(3h)\equiv3h\cdot T'_{n^{l_s}}(0)\pmod{3^2}$. 
It yields from $T'_n(0)\equiv(-1)^{\frac{n-1}{2}}n$ that $T'_{n^{l_s}}(0)=(T'_{n}(0))^{l_s}\equiv1\pmod3$. 
The term $T'_{n^{l_s}}(0)$ can be expressed as $1+3q$ for some integer $q$. 
Therefore, $T_{n^{l_s}}(s)\equiv3h\cdot (1 + 3q)\equiv s\pmod{3^2}$, establishing $v_s\geq2$.
\end{itemize}

Then, one can prove 
\begin{equation}\label{eq:intro:p2:p3}
\begin{cases}
T_{n^{l_s}}^{3^t}(s)\equiv s\pmod{3^{v_s+t}};\\
T_{n^{l_s}}^{3^t}(s)\not\equiv s\pmod{3^{v_s+t+1}}\\
\end{cases}
\end{equation}
via mathematical induction on $t$.

When $t=0$, relation \eqref{eq:intro:p2:p3} holds.
Assume that \eqref{eq:intro:p2:p3} holds for $t=e\geq0$, namely 
$T^{3^e}_{n^{l_s}}(s)=s+q_s\cdot3^{v_s+e}$, where $q_s\not\equiv0\pmod3$.
   
Let $\Phi(x)=T_{n}^{l_x\cdot 3^e}(x)$, then $T^{l_x3^{e+1}}_{n}(x)=\Phi^3(x)$ from the semi-group property of Chebyshev polynomial. 
By Taylor's formula, 
\begin{multline}\label{eq:phi2}
\Phi^2(s)=\Phi(s+q_s3^{v_s+e})=\Phi(s)+q_s3^{v_s+e}\Phi'(s)\\
+\sum_{j=2}^{\deg(\Phi)}\frac{\Phi^{(j)}(s)(q_s3^{v_s+e})^j}{j!}. 
\end{multline}
Note that $\frac{\Phi^{(j)}(s)}{j!}$ is an integer from $\Phi(x)$ is an integer coefficient polynomial, and
$(v_s+e)j\geq v_s+e+2$ for $j\geq2$.
It yields from \eqref{eq:phi2} that
$\Phi^2(s)\equiv s+q_s(1+\Phi'(s))\pmod{3^{{v_s}+e+2}}$.
Following the same procedure as in \eqref{eq:phi2}, one can obtain
\begin{equation}\label{eq:phi3}
\begin{aligned}
\Phi^3(s)&=\Phi(\Phi^2(s))\\
&\equiv s+q_s3^{v_s+e}(1+\Phi'(s)+(\Phi'(s))^2)\pmod{3^{v_s+e+2}}.  
\end{aligned}
\end{equation}

When $s\equiv\pm1\pmod3$, from \eqref{eq:T'n:+-1}, one has $T'_{n^{l_s3^e}}(s)\equiv n^{2l_s3^e}\pmod3$.
Then one further gets $T'_{n^{l_s3^e}}(s)\equiv 1\pmod3$ from $n^2\equiv1\pmod3$.
When $s\equiv0\pmod3$, $T'_{n^{l_s3^e}}(s)\equiv T'_{n^{l_s3^e}}(0) \equiv (-1)^{\frac{n^{l_s3^e}-1}{2}}n^{l_s3^e}\pmod3$.
If $l_s=1$, $(-1)^{\frac{n^{3^e}-1}{2}}n^{3^e}\equiv1\pmod3$ for $n\equiv\pm1\pmod3$.
If $l_s=2$, from $n^2\equiv1\pmod3$, one has $(-1)^{\frac{n^{2\cdot3^e}-1}{2}}n^{2\cdot3^e}\equiv1\pmod3$.
Hence, $\Phi'(s)\equiv1\pmod3$. It means $1+\Phi'(s)+(\Phi'(s))^2\equiv3\pmod9$ in \eqref{eq:phi3}.
Then, one can obtain $\Phi^3(s)\equiv s+q_s3^{v_s+e+1}\pmod{3^{v_s+e+2}}$ and \eqref{eq:intro:p2:p3} holds for $t=e+1$.
Finally, setting $s=k-v_s$ in \eqref{eq:intro:p2:p3}, then one can get the least period of the sequence $S_{3^k}(s;n)$ is $l_s\cdot 3^{k-v_s}$.
\end{proof}

\section{The Graph Structure of Chebyshev Permutation Polynomials}

This Section analyzes the graph structure of Chebyshev permutation polynomials over two distinct modular rings: Sec.~\ref{Sec:Structure:z2k} examines \(\mathbb{Z}_{2^k}\) via the parity of states, and Sec.~\ref{Sec:Structure:z3k} studies \(\mathbb{Z}_{3^k}\) through the perspective of residue classes modulo \(3\).
Let $\mathcal{G}(T_n/\mathbb{Z}_{p^k})$ denote the functional graph of the $n$-th order Chebyshev permutation polynomial over ring $\mathbb{Z}_{p^k}$, and $Cyc(L, \bullet)$ as a cycle of length $L$ in a graph.
Then, the expression $k \times Cyc(L, \bullet)$ represents the configuration of $k$ distinct components, each possessing a cycle of length $L$.
In addition, the direct operator $\bigoplus$ denotes the disjoint union of these heterogeneous topological structures within the graph.
For instance, as illustrated in Fig.~\ref{fig:SMNcat} c), the graph $\mathcal{G}(T_{19}/\mathbb{Z}_{2^5}) = \bigoplus_{i=1}^2 2 \times Cyc(2^i, \bullet)\bigoplus 20 \times Cyc(1, \bullet)$ signifies that the functional graph of the $19$-th order Chebyshev polynomial over the ring $\mathbb{Z}_{2^5}$ is composed of $20$ self-loops, two cycles of length two, and two cycles  of length four.

\subsection{The Graph Structure of Chebyshev Permutation Polynomials over Ring $\mathbb{Z}_{2^k}$}
\label{Sec:Structure:z2k}
Chebyshev permutation polynomials over the ring $\mathbb{Z}_{2^k}$ exhibit significant algebraic regularities throughout the iteration process, where the induced functional graph admits a natural decomposition into independent components with distinct topological features based on the parity of the states. 
Building upon the inherent properties of these polynomials modulo $2^k$, Proposition \ref{pro:struct:p2:odd} provides a rigorous characterization of the subgraphs constituted by odd states, specifying the exact distribution of cycle lengths and their respective multiplicities. Furthermore, Proposition \ref{pro:struct:p2:even} details the state-transition trajectories and cycle architectures that originate from even states.

\begin{proposition}\label{pro:struct:p2:odd}
As for any Chebyshev polynomial, all odd states consist of
\[
\mathcal{G}(T_n/\Zt{k}) = \bigoplus_{t=1}^{k-w-3}2^{w+1} \times Cyc(2^t, \bullet) \oplus 2^{w+2} \times Cyc(1, \bullet),
\]
if $k\geq w+3$; $\mathcal{G} =2^k \times Cyc(1, \bullet)$ otherwise.
Specially, the states of $2^{w+1}$ cycles $Cyc(2^t, \bullet)$ are
\[
 \{2^{k-w-t-1}\pm1+2^{k-w-t}j_1+2^{k-t}j_2|j_2\in A_t\},
\]
where $j_1\in A_{w}$ and $A_t=\{0, 1, \cdots, 2^t-1\}$ and
$2^{w+2}$ self-loops are $\{1+2^{k-w-1}j, (j+1)2^{k-w-1}-1|j\in A_{w+1}\}$.
\end{proposition}
\begin{proof}
For any odd state, it can be expressed as $1+2j$, where $j$ is an integer. 
By Taylor's formula and $T_n(1)=1$, one has
\begin{equation}\label{eq:1+2j}
T_n(1+2j)=1+\sum^n_{m=1}\frac{(2j)^m}{m!}\cdot T^{(m)}_n(1).
\end{equation}
Note that 
\(
\nu_2(\frac{2^m}{m!}) = \nu_2(2^m) -\nu_2(m!)= s_2(m)\ge 1
\)
from Legendre's formula.
So, one has $s_2(m)+\lfloor\frac{m}{2}\rfloor\ge 4$ when $m\ge 5$, yielding
$\nu_2(\frac{(2j)^m}{m!}\cdot T^{(m)}_n(1))\ge w+4$ from \eqref{eq:2wmid:T'n}.
It yields from \eqref{eq:1+2j} that 
\(T_n(1+2j)\equiv 1+\sum_{m=1}^{4}\frac{(2j)^m}{m!}T_n^{(m)}(1) \pmod{2^{w+4}}.
\)
By incorporating \eqref{eq:T'n:+-1} into the previous congruence
\(T_n(1+2j) 
\equiv 1+2j n^2 +2j^2\tfrac{n^2(n^2-1)}{3}
+(j^3+j^4)2^{w+3} \pmod{2^{w+4}}
\equiv 1+2j +2(n^2-1)(j+3^{-1}n^2j^2)\pmod{2^{w+4}}
\)
Note that $n^2\equiv1\pmod{2^{w+1}}$ from \eqref{eq:v2(n^2-1)} and $3^{-1}\equiv 3 \pmod 4$, one obtains
\begin{equation}\label{eq:Tn1+2j:w+4}
T_n(1+2j)\equiv 1+2j +j(1+3j)2^{w+2}\pmod{2^{w+4}}.
\end{equation}
Hence, 
\begin{equation*}
\begin{aligned}
T_n(1+2j)
&\equiv 1+2j+(j+j^2)2^{w+2}\pmod {2^{w+3}}\\
&\equiv 1+2j\pmod {2^{w+3}}.
\end{aligned}
\end{equation*}
Consequently, when $k \leq w + 3$, every odd state corresponds to a self-loop, and thus the graph can be represented as $\mathcal{G} = 2^{k} \times \mathrm{Cyc}(1, \bullet)$.
This result verifies that the proposition holds for all $k \leq w + 3$.

When $k = w+4$, referring to \eqref{eq:Tn1+2j:w+4}, one can get
$$
T_n(1+2j)\equiv
\begin{cases}
    1+2j &\mbox{if~} j\bmod 4 \in \{0, 1\};\\
    1+2j+2^{w+3} &\mbox{if~} j\bmod 4 \in \{2, 3\}.
\end{cases}
$$
This means that for all self-loop states $\{1+2j\mid j\in A_{w+2}\}$ in $G(T_n/\Zt{w+3})$, half of the states with $j\bmod4 \in\{0, 1\}$ remain self-loop,
and half of the states with $j\bmod4 \in\{2,3\}$ have their period changed in graph $G(T_n/\Zt{w+4})$.
Thus, this proposition holds for $k=w+4$.

Assume this proposition is true for $k =w+s$ with $s\geq4$, namely
\[
\mathcal{G}(T_n/\Zt{w+s})= \bigoplus_{t=1}^{s-3}2^{w+1} \times Cyc(2^t, \bullet) \oplus 2^{w+2} \times Cyc(1, \bullet).
\]
where $2^{w+2}$ self-loop cycles are 
$\{1 + 2^{s-1}j, 2^{s-1}-1 + 2^{s-1}\mid j\in A_{w+1}\}$, which yields that
$T_n(1+2^{s-1}j)\equiv 1+ 2^{s-1}j\bmod 2^{w+s}$ and 
$T_n(2^{s-1}-1+2^{s-1}j)\equiv 2^{s-1}-1+ 2^{s-1}j\bmod 2^{w+s}$.
From Lemma~\ref{pro:2w:Tn'1}, one can get
\begin{equation*}
\begin{split}
&T_n(1+2^{s-1}j)\\
&\equiv T_n(1) +2^{s-1}j\cdot T'_n(1)\\
&\equiv 1+2^{s-1}j+2^{w+s}j\\
&\equiv
\begin{cases}
1+2^{s-1}j \pmod{2^{w+s+1}} & \mbox{if~} j\equiv0 \bmod2; \\
1+2^{s-1}j+2^{w+s} \pmod{2^{w+s+1}} & \mbox{if~} j\equiv1 \bmod2,
\end{cases}
\end{split}
\end{equation*}
and 
\begin{equation*}
\begin{split}
&T_n(-1+2^{s-1} +2^{s-1}j)\\
&\equiv T_n(-1) +2^{s-1}(j+1)\cdot T'_n(1)\\
&\equiv -1+2^{s-1}(j+1)+2^{w+s}(j+1)\pmod 2^{w+s+1}
\end{split}
\end{equation*}
where $j\in A_{w+1}$.
Accordingly, depending on the parity of $j$, one obtains 
$T_n(-1+2^{s-1}+2^{s-1}j)\equiv -1+2^{s-1}+2^{s-1}j \pmod{2^{w+s+1}}$ 
when $j\equiv1\pmod2$, and 
$T_n(-1+2^{s-1}+2^{s-1}j)\equiv -1+2^{s-1}+2^{s-1}j+2^{w+s} \pmod{2^{w+s+1}}$ 
when $j\equiv0\pmod2$.
Hence, $T_{n}(1+ 2^{s}j)\equiv T_n(1+2^{s-1}j++2^{s-1}j) \equiv 1+2^{s}j \bmod 2^{w+s+1}$  and  $T_{n}(-1+2^{s} +2^{s}j)\equiv T_{n}(-1+2^{s-1} +2^{s-1}j+2^{s-1} +2^{s-1}j) \equiv -1+2^{s} +2^{s}j \bmod 2^{w+s+1}$.
and 
$T_n^2(1+2^s(2j_1+1))\equiv T_n(1+2^s(2j_1+1)+ 2^{w+s})\equiv 1+2^s(2j_1+1)\bmod 2^{w+s+1}$ with $j_1\in A_{w}$,
those are $2^{w}$ cycles with
$Cyc(2, \bullet) = \{1+2^{s-1}+2^sj_1 +2^{w+s}j_2 | j_2\in A_1\}$.
Similarly, one has $T_n^2(-1+2^{s-1}+2^sj_1))\equiv T_n(-1+2^{s-1}+2^sj_1+ 2^{w+s})\equiv -1+2^{s-1}+2^sj_1\bmod 2^{w+s+1}$,
those are $2^{w}$ cycles
$Cyc(2, \bullet) = \{-1+2^{s-1}+2^sj_1 +2^{w+s}j_2 | j_2\in A_1\}$.

Referring to \cite[Theorem 2]{Yoshioka:chebypro:Nlinear2015}, one can know cycle
$$Cyc(2^{t-1}, \bullet) =
 \{2^{s-t}\pm1+2^{s-t+1}j_1+2^{w+s-t+1}j_2|j_2\in A_{t-1}\},
$$ in $G(T_n/\Zt{w+s})$ changes cycle
$$Cyc(2^t, \bullet) =
 \{2^{s-t}\pm1+2^{s+1-t}j_1+2^{w+s+1-t}j_2|j_2\in A_t\}$$ in $G(T_n/\Zt{w+s+1})$ as the precision $k$ increases by one.
Therefore, this Property also holds for $k=w+s+1$.
\end{proof}

\begin{proposition}\label{pro:struct:p2:even}
As for any Chebyshev polynomial, all even states consist of
\[
\mathcal{G}(T_n/\Zt{k})= \bigoplus_{t=1}^{k-w-1}2^{w-1} \times Cyc(2^t, \bullet) \oplus 2^{w} \times Cyc(1, \bullet),
\]
if $k > w+1$; $\mathcal{G}(T_n/\Zt{k})=2^k\times Cyc(1, \bullet)$ otherwise.
Specifically, the \(2^{w-1}\) \(Cyc(2^t, \bullet)\) are generated from initial values \(x_0\) in the set
\[
\{2^{k-w-t}(2j_1+1) + 2^{k-t} j_2 \mid j_2 \in A_t\},
\]
where \(j_1 \in A_{w-1}\) and  \(A_i = \{0, 1, \cdots, 2^i-1\}\). Moreover, the \(2^w\) self-loops are generated from initial values \(x_0\) in the set \(\{2^{k-w} j \mid j \in A_w\}\).
\end{proposition}
\begin{proof}
Note that $T_n(2j)= a_1\cdot 2j+a_3(2j)^3+\cdots+a_n(2j)^n$.
It yields from Lemma~\ref{pro:2w:ak} that 
\begin{equation}\label{eq:2j:w+1}
T_n(2j)\equiv 2j\bmod 2^{w+1} 
\end{equation}
for all $j \in\{0, 1, \cdots, 2^w-1\}$. 
Thus, all even states form self-loops when $k \leq w+1$, which is represented as $\mathcal{G}(T_n/\Zt{k}) = 2^{k} \times \mathrm{Cyc}(1, \bullet)$.

When $k=w+2$, for any $j\in A_{w+1}$, Lemma~\ref{pro:2w:ak} implies
$T_n(2j)\equiv 2j(2^w\pm1+h\cdot2^{w+1})\bmod 2^{w+2}\equiv 2j+ 2^{w+1}j\bmod 2^{w+2}$.
Hence,
\begin{equation*}
T_n(2j)\equiv 
\begin{cases}
2j \bmod 2^{w+2} & \mbox{when~} j\equiv0 \bmod 2; \\
2j+2^{w+1} \bmod 2^{w+2} & \mbox{when~} j\equiv1 \bmod 2.
\end{cases}
\end{equation*}
Applying $T_n$ once again yields
\begin{equation*}
\begin{multlined}
T_n^2(2j) \equiv \\
\begin{cases}
T_n(2j) \pmod{2^{w+2}}           & \text{when } j \equiv 0 \pmod{2}; \\
T_n(2j + 2^{w+1}) \pmod{2^{w+2}} & \text{when } j \equiv 1 \pmod{2}.
\end{cases}
\end{multlined}
\end{equation*}
In both cases, $T_n^2(2j)\equiv 2j\pmod{2^{w+2}}$.
The two congruences above imply that all states of the form $x = 2j$ are self-loops, where $j \in A_{w+1}$ and $j \equiv 0 \pmod{2}$, they form $2^w$ self-loops, that is $2^w \times \mathrm{Cyc}(1, \bullet)$. Meanwhile, all states of the form $x = 2j$ constitute $\frac{2^{w+1} - 2^w}{2} = 2^{w-1}$ cycles of length $2$, where $j \in A_{w+1}$ and $j \equiv 1 \pmod{2}$, expressed as $2^{w-1} \times \mathrm{Cyc}(2, \bullet)$. This confirms that the proposition holds for $k = w + 2$.
Assume the Proposition holds for any $k =w+s$, namely
\[
\mathcal{G}(T_n/\Zt{w+s}) = \bigoplus_{t=1}^{s-1}2^{w-1} \times Cyc(2^t, \bullet) \oplus 2^{w} \times Cyc(1, \bullet),
\]
where $2^w$ self-loop cycles are $\{2^{s}j|j\in A_w\}$, which yields that
$T_n(2^sj)\equiv 2^sj\bmod 2^{w+s}$.
Similar to \eqref{eq:2j:w+1}, one has
\begin{equation*}
\begin{split}
T_n(2^sj)& \equiv 2^sj (2^w\pm1+h\cdot2^{w+1})\bmod 2^{w+s+1}\\
&\equiv 2^sj+ 2^{w+s}j\bmod 2^{w+s+1}\\
&\equiv
\begin{cases}
 2^sj \bmod 2^{w+s+1} & \mbox{if~} j\equiv0 \bmod2; \\
 2^sj+ 2^{w+s} \bmod 2^{w+s+1} & \mbox{if~} j\equiv1 \bmod2.
\end{cases}
\end{split}
\end{equation*}
It further follows from the above equation that  
\[
T_n^2(2^s j) \equiv 2^s j \pmod{2^{w+s+1}}.
\]
This implies that the states $x \in \{2, 4, 6 ,\ldots, 2^w-2\}$ are self-loops, 
while the states $x \in \{1, 3, 5, \ldots, 2^w-1\}$ form cycles of length two. 
Referring to \cite[Theorem~2]{Yoshioka:chebypro:Nlinear2015}, we know that the cycle
\begin{equation*}
Cyc(2^{t-1}, \bullet)
= \{2^{s-t+1}j_1 + 2^{w+s-t+1}j_2 
\mid j_1, j_2 \in A_t\}
\end{equation*}
in $G(T_n/\mathbb{Z}_{2^{w+s}})$ transforms into the cycle
\begin{equation*}
Cyc(2^{t}, \bullet) 
= \{2^{s+1-t}j_1 + 2^{w+s+1-t}j_2 
\mid j_1, j_2 \in A_t\}
\end{equation*}
in $\mathcal{G}(T_n/\mathbb{Z}_{2^{w+s+1}})$ as the precision $k$ increases by one, where $j_1 \equiv 1 \pmod{2},
j_2 \equiv 0 \pmod{2}$.
Therefore, this property also holds for $k = w + s + 1$.
\end{proof}

Figure \ref{fig:SMNcat} illustrates the evolution of the functional graph of the Chebyshev polynomial of degree $19$ over the ring $\mathbb{Z}_{2^k}$ as $k$ continuously increases.
Combining Propositions \ref{pro:struct:p2:odd} and \ref{pro:struct:p2:even}, it can be calculated that $w=2$, then one has all odd states consist of
\[
\mathcal{G}(T_n/\Zt{7}) = \bigoplus_{t=1}^{2}8 \times Cyc(2^t, \bullet) \bigoplus 16 \times Cyc(1, \bullet);
\]
all even states consist of
\[
\mathcal{G}(T_n/\Zt{7}) = \bigoplus_{t=1}^{4}2 \times Cyc(2^t, \bullet) \bigoplus 4 \times Cyc(1, \bullet).
\]
Therefore, one obtains the functional graph
\begin{multline*}
\mathcal{G}(T_n/\mathbb{Z}_{2^7}) = \bigoplus^4_{t=3} 2\times Cyc(2^t, \bullet) \bigoplus^2_{t=1} 10\times Cyc(2^t, \bullet) \\ \bigoplus 20\times Cyc(1, \bullet),
\end{multline*}
The structure exhibited in Figure \ref{fig:SMNcat} e) is consistent with the deduced structure.

\renewcommand\arraystretch{1.2}
\setlength\tabcolsep{4pt} 
\addtolength{\abovecaptionskip}{-2pt}
\begin{figure*}[!htb]
	\centering
	\begin{minipage}[t]{0.4\twofigwidth}
	\centering
	\includegraphics[width=0.56\twofigwidth]{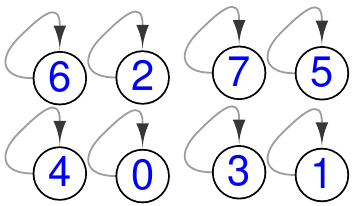}
	a)
	\end{minipage}\hspace{6em}
	\begin{minipage}[t]{0.5\twofigwidth}
		\centering
		\includegraphics[width=0.5\twofigwidth]{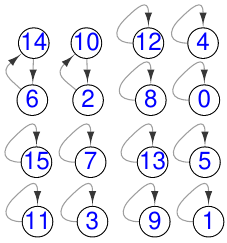}
		b)
	\end{minipage} \hspace{6em}
	\begin{minipage}[t]{\twofigwidth}
		\centering
		\includegraphics[width=\twofigwidth]{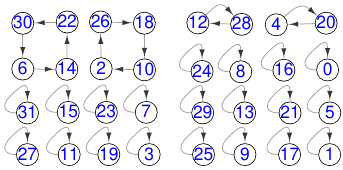}
		c)
	\end{minipage}\vspace{0.2em}
	\begin{minipage}{0.8\BigOneImW}
		\centering
		\includegraphics[width=0.8\BigOneImW]{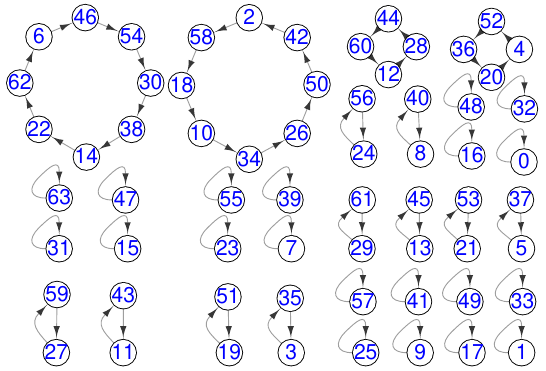}
		d)
	\end{minipage}
	\begin{minipage}{0.95\BigOneImW}
		\centering
		\includegraphics[width=0.95\BigOneImW]{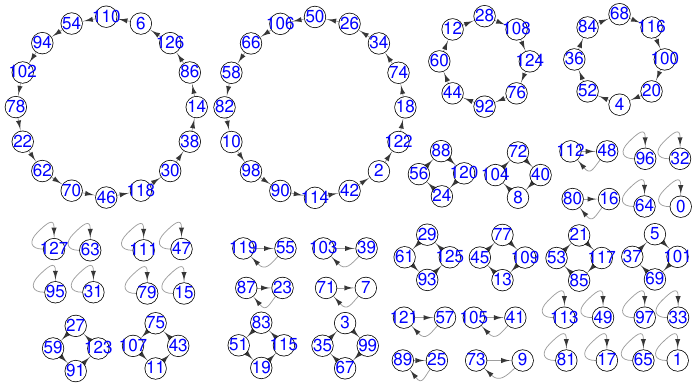}
		e)
	\end{minipage}
	\caption{Functional graphs of Chebyshev polynomial  $(n, p)=(19, 2)$:
	a) $e=3$; b) $e=4$; c) $e=5$; d) $e=6$; c) $e=7$.}
\label{fig:SMNcat}
\end{figure*}

\begin{corollary}\label{theorem:p2:period:express}
Let $L(n, 2, k, x)$ denote the period of the $n$-th order Chebyshev permutation polynomial at $x\in\Zt{k}$. Then
\begin{equation*}
L(n, 2, k, x)=
\begin{cases}
2^{k'-\bnu_2(x)}     & \text{when }x \equiv 0\!\!\! \pmod{2}, k'>1,\\[2pt]
2^{k'-\bnu^*_2(x)-1} & \text{when }x \equiv 1\!\!\! \pmod{2}, k'>3,\\[2pt]
1                    & \text{otherwise},
\end{cases}
\end{equation*}
where $k' = k -\bnu^*_2(n)$, $\bnu_2(x)$ denotes the $2$-adic valuation of $x$, and
\[
\bnu^*_2(x) = \max\{\bnu_2(x-1), \bnu_2(x+1)\}.
\]
\end{corollary}
\begin{proof}
Based on the functional graph structure Propositions~\ref{pro:struct:p2:odd} and~\ref{pro:struct:p2:even} concerning odd and even states, we can derive the result of this corollary.
\end{proof}


\subsection{The Graph Structure of Chebyshev Permutation Polynomials over Ring $\mathbb{Z}_{3^k}$}\label{Sec:Structure:z3k}

Based on the distribution of initial states among modulo 3 residue classes, as described in Theorem \ref{Theorem:periodValue}, the evolution paths and dynamical characteristics of various vertices within the functional graph are thoroughly examined in this section.

Proposition \ref{lemma:Rk:self-loop} further precisely locates the distribution of self-loops within the functional graph. 
Combining the period distribution with the topological features of self-loops, Proposition \ref{pro:SMN:p3:0} comprehensively characterizes the functional graph for the case $x \equiv 0$, detailing the lengths and corresponding cycles. 
Analogously, Proposition \ref{pro:SMN:p3:pm1} details the functional graph for the case of $x \equiv \pm 1$. 
Due to the methodological similarity in the derivations, the proof of Proposition \ref{pro:SMN:p3:pm1} is omitted.

\begin{figure*}[!htb]
    \centering
	\begin{minipage}[t]{2.3\twofigwidth}
	\centering
	\includegraphics[width=2.3\twofigwidth]{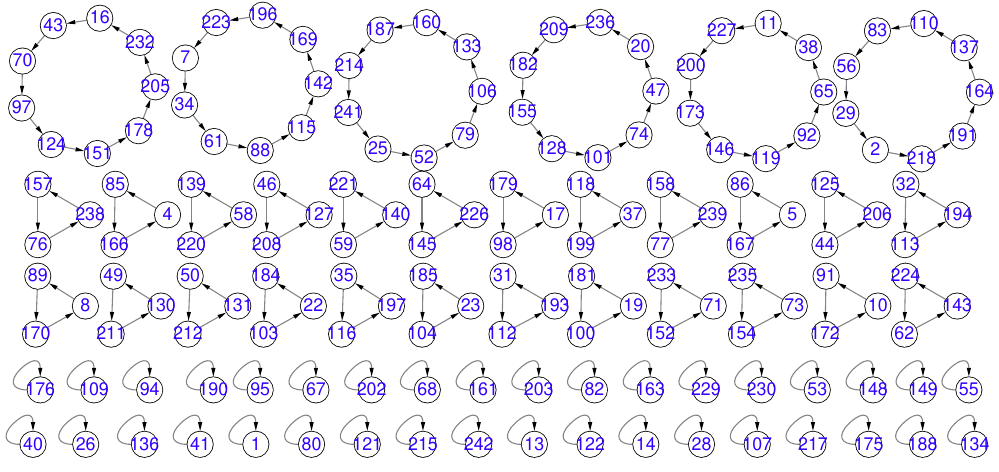}
	a)
	\end{minipage}
	\begin{minipage}[t]{1.3\twofigwidth}
		\centering
		\includegraphics[width=1.3\twofigwidth]{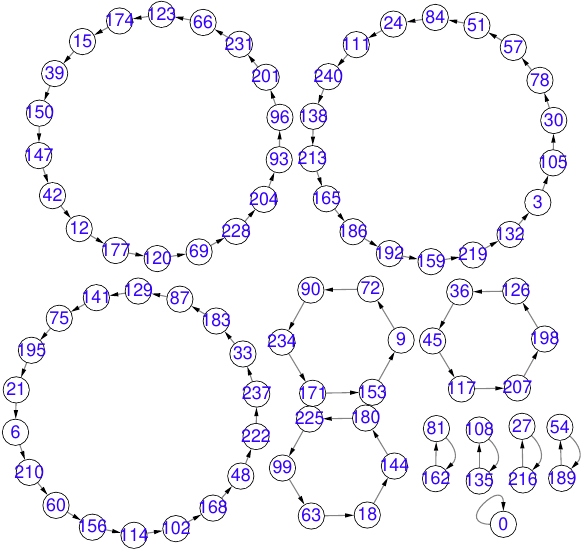}
		b)
	\end{minipage} \hspace{6em}
	\caption{Functional graphs of Chebyshev polynomial $(n, p, k)=(19, 3, 5)$ with all status satisfying: a) $x\equiv0\pmod3$ b) $x\equiv\pm1\pmod3$.}
\label{fig:cheby:Z3p:n19:k2345}
\end{figure*}

\begin{proposition}\label{lemma:Rk:self-loop}
Let $R_k$ be the set of values of self-loops in $\mathcal{G}_{3^k}$ whose values belong to $A_1$.
Then, one has
$R_k=A_1$ if $k\leq w+1$;
otherwise 
\[
R_k = \{x\in A_1 \vvert x\equiv (\pm 1 + \sum_{i=1}^{k-w-1} j_i 3^i) \pmod{3^{k-w}},\}
\]
where $A_1=\{x\vert x\in\mathbb{Z}_k, x\equiv \pm 1\pmod3\}$,
\begin{equation}\label{eq:j_i:condtion}
j_i\in 
\begin{cases}
\{0, 1\} & \mbox{when } i=1, \\
\{-q_i(1+j_1)^{-1}\bmod 3 \}  & \mbox{when } 2 \le i \leq k-w-1,
\end{cases}
\end{equation}
and
$q_i=\frac{T_n(\pm1+\sum^{i-1}_{t=1}3^tj_t)\mp1+\sum^{i-1}_{t=1}3^tj_t}{3^{w+i}}$.
\end{proposition}
\begin{proof}
Since $x\equiv\pm 1\pmod 3$ and $x\in \Zy{k}$, one can set $x = (\pm1 + 3h) \bmod{3^k}$, where $h = \sum^{k-1}_{i=1}j_i3^{i-1}$, $j_i\in\{0, 1, 2\}$ for $1\leq i\leq k-1$.
Applying Taylor's formula, one obtains
\begin{equation}\label{eq:pm1+3t:taylor}
T_n(\pm 1+3h)=\pm 1 + 3hT'_n(\pm1) + \cdots + \frac{(3h)^n}{n!}T^{(n)}_{n}(\pm1).
\end{equation}
From Lemma~\ref{lemma:p3:Tn'1}, one has 
$\frac{(3h)^m}{m!}\cdot T_n^{(m)}(\pm1)\equiv0\pmod{3^{w+2}}$ for any $m\geq3$.
It yields from \eqref{eq:pm1+3t:taylor} that
$T_n(\pm 1+3h) \equiv \pm 1+3hT'_n(\pm1) + \frac{h^23^2T''_n(\pm1)}{2}\pmod{3^{w+2}}$.
Then, according to \eqref{eq:T'n:+-1}, one has 
\begin{equation}\label{eq:pm1+3t:taylor:futher}
T_n(\pm 1+3h) \equiv \pm1+3h+3h(n^2-1)(1+2^{-1}an^2)\pmod{3^{w+2}}.
\end{equation}
Note that $t\equiv j_1 \pmod3$ and $(1+2^{-1}hn^2)\equiv (1+2j_1) \pmod3$. 
Substituting the two congruences and $\bnu_3(n^2-1)=w$ into \eqref{eq:pm1+3t:taylor:futher}, one gets
\begin{equation*}
T_n(\pm1+3h)\equiv
\begin{cases}
\pm1+3h+3^{w+1}\cdot j_1(1+2j_1)\pmod{3^{w+2}}; \\ \pm1+3h\pmod{3^{w+1}}.
\end{cases}
\end{equation*}
Thus, $R_k=A_1$ when $k\leq w+1$, and
$R_{w+2}=\{x\in\mathbb{Z}_{3^k} \vvert x \equiv\pm1+3j_1\pmod9\}$, where $j_1\in\{0, 1\}$.

Suppose the proposition holds for $k=w+s$ with $s\geq2$, one has 
\(
T_n(\pm 1 + \sum_{i=1}^{s-1} j_i 3^i) \equiv \pm 1 + \sum_{i=1}^{s-1}j_i3^i\pmod{3^{w+s}},
\)
where $j_i$ satisfies \eqref{eq:j_i:condtion}.
Denote $\Delta=\sum^{s-1}_{i=1}3^ij_i$, one has 
\begin{equation}\label{eq:pm1:delta:q}
T_n(\pm 1+\Delta) =\pm 1+\Delta+q_{s}\cdot3^{w+s}.  
\end{equation}
When $k=w+s+1$, applying Taylor’s formula obtains
\begin{multline}\label{eq:pm1+delta+B}
T_n(\pm1+\Delta+3^sB_s)=\\T_n(\pm1+\Delta)+\sum^n_{m=1}\frac{(3^sB_s)^m}{m!}T_n^{(m)}(\pm1+\Delta),
\end{multline}
where $B_s=\sum^k_{i=s}j_i3^{i-s}$.
Note that
$\frac{T_n^{(m)}(\pm1+\Delta)\cdot 3^m}{m!}=\frac{T_n^{(m)}(\pm1)\cdot 3^m}{m!}+\sum_{i=m+1}^n\frac{\Delta^{i-m}\cdot T_n^{(i)}(\pm1)\cdot3^m}{(i-m)!\cdot m!}
=\frac{T_n^{(m)}(\pm1)\cdot3^m}{m!}+\sum_{i=m+1}^n(\frac{\Delta}{3})^{i-m}\binom{i}{m} \frac{T_n^{(i)}(\pm1)\cdot3^i}{i!}$.
It yields from Lemma \ref{lemma:p3:Tn'1} that
\begin{multline}\label{eq:pm1+3h:m}
\frac{T^{(m)}_n(\pm 1+\Delta) \cdot 3^m}{m!}\equiv\\
\begin{cases}
0\pmod{3^{w+2}}  &\mbox{if~} m\ge 3;\\
\frac{T^{''}_n(\pm 1) \cdot 3^2}{2}\pmod{3^{w+2}}  &\mbox{if~} m=2;\\
3T'_n(\pm 1)+3\Delta T''_n(\pm1)  \pmod{3^{w+2}}  &\mbox{if~} m=1.
\end{cases}
\end{multline}
So, $\frac{(3^sB_s)^m}{m!}T^{(m)}_n(\pm1+\Delta)\equiv0\pmod{3^{w+s+1}}$ for any $m\geq3$, $\frac{(3^sB_s)^2T''_n(\pm1+\Delta)}{2}\equiv \frac{(3^sB_s)^2\cdot T''_n(\pm1)}{2}\pmod{3^{w+s+1}}$, and
$3^sB_sT'_n(\pm1+\Delta)\equiv 3^sT'_n(\pm1)+3^s\Delta T''_n(\pm1)\pmod{3^{w+s+1}}$.
Substituting the previous three congruences into \eqref{eq:pm1+delta+B}, one gets 
\(
T_n(\pm1+\Delta+3^sB_s)\equiv T_n(\pm1+\Delta)+3^sB_s(T'_n(\pm1)+\Delta T''_n(\pm1))+\frac{(3^sB_s)^2\cdot T''_n(\pm1)}{2}\pmod{3^{w+s+1}}.
\)
It yields from \eqref{eq:T'n:+-1} that 
$T_n(\pm1+\Delta+3^sB_s)\equiv
T_n(\pm1+\Delta)+3^sB_sn^2+3^{s-1}\Delta B_sn^2(n^2-1)+
3^{2s-1}B_s^2\frac{n(n^2-1)}{2}  
\pmod{3^{w+s+1}}$.
Then, combining $\bnu_3(n^2-1)=w$, one further gets
\(
T_n(\pm1+\Delta+3^sB_s)\equiv T_n(\pm1+\Delta)+3^sB_s+
3^{w+s}(j_s+j_sj_1)\pmod{3^{w+s+1}}.
\)
Combining \eqref{eq:pm1:delta:q}, one can get 
\begin{multline*}
T_n(\pm1+\Delta+3^sB_s)\equiv \pm1+\Delta+3^sB_s+\\3^{w+s} (q_s+j_s+j_sj_1)\pmod{3^{w+s+1}}.
\end{multline*}
So, when $(q_s+j_s+j_sj_1)\equiv 0\pmod 3$, namely $j_s =-q_s(1+j_1)^{-1}\bmod 3$, one has 
\begin{multline*}
T_n(\pm1+\Delta+3^sB_s)\equiv \pm1+\Delta+3^sB_s\pmod{3^{w+s+1}}
\end{multline*}
for any  $j_i\in \{0,1,2\}$ when $i\ge s+1$.
So, when 
$x\equiv (\pm 1 + \sum_{i=1}^{s} j_i 3^i) \pmod{3^{s+1}}$, one has $T_n(x)\equiv x\pmod{3^{w+s+1}}$,
which means the proposition holds for $k=w+s+1$.
\end{proof}

\begin{proposition}\label{pro:SMN:p3:0}
The connected components of $\mathcal{G}_{3^k}$ to which all states $x$ satisfying $x\equiv0\pmod 3$ belong, are composed of
$\frac{2\cdot 3^{w-1}}{l_0}$ cycles of length $l_03^i$, $\frac{3^w-1}{l_0}$ cycles of length $l_0$, and one self-loop when $k\ge w+2$;
$\frac{3^{k-1}-1}{l_0}$ cycles of length $l_0$ and one self-loops otherwise, where $l_0=\ord(T'_n(0))$, $i\in\{1, 2, \cdots, k-w-1\}$.
\end{proposition}
\begin{proof}
When $x=0$, one has $T_n(0)=0$, which means the node with status value $0$ is a self-loop in $\mathcal{G}_{3^k}$.
When $x\equiv0\pmod3$ and $x\neq0$, one can write $x=3t$, where $t$ is a positive integer.
By Taylor's formula, one has 
\begin{equation}\label{eq:taylor:mod3:0}
T_n(0+3t)=T_n(0)+3tT'_n(0)+\sum_{m=2}^n\frac{(3t)^m\cdot T^{(m)}_n(0)}{m!}.
\end{equation}
According to the value of $k$, the proof can be divided into the following two cases:

\begin{itemize}
\item $k\leq w+1$:
From Lemma~\ref{lemma:dev:m:0mod3:w+2} and \eqref{eq:taylor:mod3:0}, one obtains
$T_n(3t)\equiv3t\cdot T'_n(0)\pmod{3^k}$.
It means $T_n(x)\equiv x\cdot T'_n(0)\pmod{3^k}$.
Hence, $T^i_n(x)=T_n(T^{i-1}_n(x))\equiv x\cdot [T'_n(0)]^i\pmod{3^k}$.
From Lemma~\ref{lemma:dev:m:0mod3:w+2}, one has $T'_{n^{l_0}}(0)=[T'_n(0)]^{l_0}\equiv1\pmod3^k$, which further gets $T^{l_0}_n(x)\equiv x\pmod{3^k}$.
Therefor, from the definition of $l_0$, one can get $l_0$ is the least integer such $T^{l_0}_n(x)\equiv x\pmod{3^k}$. Hence, there are $\frac{3^{k-1}-1}{l_0}$ cycles of length $l_0$ in $\mathcal{G}_{3^k}$.
\item $k\geq w+2$:
If $\bnu_3(t)\geq k-w-1$, i.e. $k\leq\bnu_3(t)+w+1$, it yields from  \eqref{eq:taylor:mod3:0} that $T_n(3t)\equiv3t\cdot T'_n(0)\pmod{3^k}$. Then from the above case, one obtains there are  $\frac{3^w-1}{l_0}$ cycles of length $l_0$ in $\mathcal{G}_{3^k}$. 
If $\bnu_3(t)\leq k-w-2$, one has $v_s=w+1+\bnu_3(t)$ in Theorem~\ref{Theorem:periodValue}. Then, the least period of the sequence is $l_03^{k-w-1-\bnu_3(t)}$. Let $s = \bnu_3(t)$, one obtains the number of cycles of length $l_03^i$ is
\begin{equation*}
\frac{|\{t \in \mathbb{Z}_{3^{k-1}} \mid \nu_3(t)=s\}|}{l_03^{k-w-1-s}} = \frac{2 \cdot 3^{k-2-s}}{l_03^{k-w-1-s}} = \frac{2 \cdot 3^{w-1}}{l_0},
\end{equation*}
where $i=k-w-1-s \in \{1, 2, \ldots, k-w-1\}$.
\end{itemize}
\end{proof}

\begin{proposition}\label{pro:SMN:p3:pm1}
The connected components of $\mathcal{G}_{3^k}$ to which all states $x$ satisfying $x\equiv\pm1\pmod 3$ belong, are composed of $8\cdot 3^{w-1}$ cycles of length $3^i$, $4\cdot 3^w$ self-loops, and $2\cdot 3^{w-1}$ cycles of length $3^{k-w-1}$ when $k\ge w+2$; and $2\cdot 3^{k-1}$ self-loops otherwise, where $i\in\{1, 2, \cdots, k-w-2\}$.
\end{proposition}

As a typical example with $n=19$, it is determined that $w=2$ and $l_0=2$. According to Propositions~\ref{pro:SMN:p3:0} and \ref{pro:SMN:p3:pm1}, all states $x\equiv\pm1\pmod3$ form the graph $\mathcal{G}(T_n/\mathbb{Z}_{3^5}) = 24 \times cyc(3, \bullet) \bigoplus 36 \times cyc(1, \bullet) \bigoplus 6 \times cyc(9, \bullet)$, whereas all states $x\equiv0\pmod3$ constitute the graph $\mathcal{G}(T_n/\mathbb{Z}_{3^5}) = 1 \times cyc(1, \bullet) \bigoplus 4 \times cyc(2, \bullet) \bigoplus_{i=1}^{2} 3 \times cyc(2 \cdot 3^i, \bullet)$. The functional graph structures illustrated in Fig.~\ref{fig:cheby:Z3p:n19:k2345} a) and b) demonstrate complete consistency with Propositions \ref{pro:SMN:p3:pm1} and \ref{pro:SMN:p3:0}, respectively.

\section{Conclusion}

This paper analyzed the functional graph induced by Chebyshev permutation polynomials over the composite ring $\mathbb{Z}_{2^{k_1}3^{k_2}}$. By leveraging new structural properties of Chebyshev polynomials modulo powers of $2$ and $3$, we established explicit characterizations of both the path lengths and the cycle structures of the resulting dynamical systems. Despite the domain's mixed binary-ternary nature, our results reveal that the functional graphs exhibit strong and predictable regularities. In particular, the number of fixed-length cycles remains constant under parameter scaling, and the associated branching patterns evolve in a controlled and analyzable manner as $k_1$ and $k_2$ increase.
These findings extend existing studies on Chebyshev dynamics over prime-power rings to the broader setting of composite rings. Beyond their intrinsic number-theoretic interest, the results also shed light on how complexity emerges in digital nonlinear maps, thereby providing helpful theoretical support for evaluating the pseudo-randomness and security properties of Chebyshev-based constructions in cryptographic applications.

\bibliographystyle{IEEEtran_doi}
\bibliography{Graph_Chebyshev}

@inproceedings{Umeno:Permute,
  author =       {Ken Umeno},
  title =        {Key exchange by {C}hebyshev polynomials modulo $2^w$},
  booktitle =    {Proceedings of Indonesia Cryptology and Information Security},
  year =         {2005},
  pages =        {95--97},
}

@article{Liao:ITC:2010,
  author =       {Xiaofeng Liao and Fei Chen and Kwok-Wo Wong},
  title =        {On the security of public-key algorithms based on {C}hebyshev polynomials over the finite field $\mathbb{Z}_N$},
  journal =      {IEEE Transactions on Computers},
  year =         {2010},
  volume =       {59},
  number =       {10},
  pages =        {1392--1401},
  doi =          {10.1109/TC.2010.148},
}

@article{Chen:chebyZN:IS2011,
  author =       {Fei Chen and Xiaofeng Liao and Tao Xiang and Hongying Zheng},
  title =        {Security analysis of the public key algorithm based on {C}hebyshev polynomials over the integer ring $\mathbb{Z}_N$},
  journal =      {Information Sciences},
  year =         {2011},
  volume =       {181},
  number =       {22},
  pages =        {5110-5118},
  doi =          {10.1016/j.ins.2011.07.008},
}

@article{bulychev:Chebyshev:ARC2003,
  author =       {Bulychev, Yu G and Bulycheva, E Yu},
  title =        {Some new properties of the {C}hebyshev polynomials and their use in analysis and design of dynamic systems},
  journal =      {Automation and Remote Control},
  year =         {2003},
  volume =       {64},
  number =       {4},
  pages =        {554--563},
  doi =          {10.1023/A:1023286213571},
}

@article{Bergamo:PKIchebyshev:CASI2005,
  author =       {P Bergamo and P D'Arco and A De Santis and L Kocarev},
  title =        {Security of public-key cryptosystems based on {C}hebyshev polynomials},
  journal =      {IEEE Transactions on Circuits and Systems I: Regular Papers},
  year =         {2005},
  volume =       {52},
  number =       {7},
  pages =        {1382--1393},
  doi =          {10.1109/TCSI.2005.851701},
}

@article{Yoshioka:ChebyshevPk:TCAS2:2018,
  author =       {Yoshioka, Daisaburo},
  title =        {Properties of {C}hebyshev polynomials modulo $p^k$},
  journal =      {IEEE Transactions on Circuits and Systems II: Express Briefs},
  year =         {2018},
  volume =       {65},
  number =       {3},
  pages =        {386--390},
  doi =          {10.1109/TCSII.2017.2739190},
}

@article{cqli:network:TCASI2019,
  author =       {Chengqing Li and Bingbing Feng and Shujun Li and Ju\"ergen Kurths and Guanrong Chen},
  title =        {Dynamic analysis of digital chaotic maps via state-mapping networks},
  journal =      {IEEE Transactions on Circuits and Systems I: Regular Papers},
  year =         {2019},
  volume =       {66},
  number =       {6},
  pages =        {2322--2335},
  doi =          {10.1109/TCSI.2018.2888688},
}

@ARTICLE{Yoshioka:ChebyshevPk:TCAS2:2020,
  author =       {D. {Yoshioka}},
  title =        {Security of public-key cryptosystems based on {C}hebyshev polynomials over $\mathbb{Z}/p^{k}\mathbb{Z}$},
  journal =      {IEEE Transactions on Circuits and Systems II: Express Briefs},
  year =         {2020},
  volume =       {67},
  number =       {10},
  pages =        {2204-2208},
  doi =          {10.1109/TCSII.2019.2954855},
}

@article{Yoshioka:chebypro:Nlinear2015,
  author =       {Daisaburo Yoshioka and Yuta Dainobu},
  title =        {On some properties of {C}hebyshev polynomial sequences modulo $2^k$},
  journal =      {Nonlinear Theory and Its Applications},
  year =         {2015},
  volume =       {6},
  number =       {3},
  pages =        {443-452},
  doi =          {10.1587/nolta.6.443},
}

@article{chenf:cat2:TIT13,
  author =       {Chen, Fei and Wong, Kwok-Wo and Liao, Xiaofeng and Xiang, Tao},
  title =        {Period distribution of the generalized discrete {A}rnold {C}at map for $n=2^e$},
  journal =      {IEEE Transactions on Information Theory},
  year =         {2013},
  volume =       {59},
  number =       {5},
  pages =        {3249-3255},
  doi =          {10.1109/TIT.2012.2235907},
}

@article{licq:Logistic:IJBC2023,
  author =       {Xiaoxiong Lu and Eric Yong Xie and Chengqing Li},
  title =        {Periodicity Analysis of {L}ogistic Map over Ring $\mathbb{Z}_{3^n}$},
  journal =      {International Journal of Bifurcation and Chaos},
  year =         {2023},
  volume =       {33},
  number =       {5},
  pages =        {art. no. 2350063},
  doi =         {10.1142/S0218127423500633},
}

@article{cqli:Cheby:TIT25, 
  author =       {Chengqing Li and Xiaoxiong Lu and Kai Tan and Guanrong Chen},
  title =        {The Graph Structure of {C}hebyshev Permutation Polynomials over Ring $\mathbb{Z}_{p^k}$}, 
  journal =      {IEEE Transactions on Information Theory},
  year =         {2025},
  volume =       {71},
  number =       {2},
  pages =        {1419--1433},
  doi =          {10.1109/TIT.2024.3522095},
}

@article{Daniel:Chebyshev:2018,
  author =       {Claudio Qureshi and Daniel Panario},
  title =        {The graph structure of {C}hebyshev polynomials over finite fields and applications},
  journal =      {Designs, Codes and Cryptography},
  year =         {2019},
  volume =       {87},
  pages =        {393--416},
  doi =          {10.1007/s10623-018-0545-7},
}

@article{TXJ:imageEncryption:2014,
  author =       {Tong, Xiao-Jun and Zhang, Miao and Wang, Zhu and Liu, Yang},
  title =        {A image encryption scheme based on dynamical perturbation and linear feedback shift register},
  journal =      {Nonlinear Dynamics},
  year =         {2014},
  volume =       {78},
  number =       {3},
  pages =        {2277-2291},
  doi =          {10.1007/s11071-014-1564-1},
}

@article{Tan:random:ND2024,
  author =       {Tan, Songyuan and Sun, Jingru and Tang, Yiping and Sun, Yichuang and Wang, Chunhua},
  title =        {Hyperchaotic bilateral random low-rank approximation random sequence generation method and its application on compressive ghost imaging},
  journal =      {Nonlinear Dynamics},
  year =         {2024},
  volume =       {112},
  number =       {7},
  pages =        {5037-5052},
  doi =          {10.1007/s11071-024-09317-0},
}

@article{Yoshioka:TCSII:2016,
  author =       {Yoshioka, Daisaburo and Kawano, Kento},
  title =        {Periodic Properties of Chebyshev Polynomial Sequences Over the Residue Ring $\mathbb{Z}/2^{k}\mathbb{Z}$},
  journal =      {IEEE Transactions on Circuits and Systems II: Express Briefs},
  year =         {2016},
  volume =       {63},
  number =       {8},
  pages =        {778-782},
  doi =          {10.1109/TCSII.2016.2531058},
}

@inproceedings{Yoshioka:isit:2023,
  author =       {Yoshioka, Daisaburo},
  title =        {Periodic Properties of Commutative Polynomials Defined by Fourth Order Recurrence Relations with Two Variables Over $\mathbb{Z}_{2^k}$},
  booktitle =    {IEEE International Symposium on Information Theory},
  year =         {2023},
  pages =        {1425-1429},
  doi =          {10.1109/ISIT54713.2023.10206860},
}

@book{knuth1997vol2,
  title     = {The Art of Computer Programming, Volume 2: Seminumerical Algorithms},
  author    = {Knuth, Donald E.},
  year      = {1997},
  edition   = {3rd},
  publisher = {Addison-Wesley Professional},
  address   = {Reading, Massachusetts},
  isbn      = {0-201-89684-2}
}

@INPROCEEDINGS{kocarev:2003:public,
  author =       {Kocarev, L. and Tasev, Z.},
  title =        {Public-key encryption based on Chebyshev maps},
  booktitle =    {Proceedings of the 2003 International Symposium on Circuits and Systems, 2003. ISCAS '03.},
  year =         {2003},
  volume =       {3},
  pages =        {III-III},
  doi =          {10.1109/ISCAS.2003.1204947},
}

@article{wang:perturbation:TCAS2016,
  author =       {Qianxue Wang and Simin Yu and Chengqing Li and Jinhu L{\"{u}} and Xiaole Fang and Christophe Guyeux and Jacques M. Bahi},
  title =        {Theoretical Design and {FPGA}-Based Implementation of Higher-Dimensional Digital Chaotic Systems},
  journal =      {IEEE Transactions on Circuits and Systems I: Regular Papers},
  year =         {2016},
  volume =       {63},
  number =       {3},
  pages =        {401--412},
  doi =          {10.1109/TCSI.2016.2515398}
}

@misc{panraksa:fixedChebysev:arxiv2026,
  author =       {Chatchawan Panraksa and Aram Tangboonduangjit},
  title =        {Fixed-point lifting and ghost periodic points for Chebyshev polynomials modulo odd prime powers},
  eprint =       {2605.04417},
  howpublished = {},
  year =         {2026},
  archivePrefix= {arXiv},  
  doi =          {10.48550/arXiv.2605.04417},
}

\end{document}